\documentclass[11pt,a4paper,reqno]{amsart}

\usepackage{cite}
\usepackage{amsmath,amssymb,amsfonts,amsthm}
\usepackage{algorithmic}
\usepackage{graphicx}
\usepackage{textcomp}
\usepackage[dvipsnames]{xcolor}
\numberwithin{equation}{section}

\theoremstyle{definition}
\newtheorem{theorem}{Theorem}[section]
\newtheorem{corollary}[theorem]{Corollary}
\newtheorem{proposition}[theorem]{Proposition}
\newtheorem{definition}[theorem]{Definition}

\newtheorem{notation}[theorem]{Notation}
\newtheorem{remark}[theorem]{Remark}
\newtheorem{lemma}[theorem]{Lemma}

\usepackage[margin=3.2cm]{geometry}

\usepackage{calrsfs}

\newcommand{\numberset}{\mathbb}

\newcommand{\Z}{\numberset{Z}}

\newcommand{\R}{\numberset{R}}
\newcommand{\C}{\mathcal{C}}

\newcommand{\F}{\numberset{F}}

\newcommand{\HH}{\textnormal{H}}

\newcommand{\mC}{\mathcal{C}}

\newcommand{\dH}{d^{\textnormal{H}}}
\newcommand{\wH}{\omega^{\textnormal{H}}}

\newcommand\red[1]{{{\textcolor{red}{#1}}}}

\usepackage{hyperref}

\title[]{\large{\textbf{Duality and LP Bounds for Codes with Locality}}}
\usepackage{authblk}

\usepackage[foot]{amsaddr}
\author{Anina Gruica$^1$}
\address{$^1$Eindhoven University of Technology, the Netherlands.}
\thanks{$^1$A. G. is supported by the Dutch Research Council through grant OCENW.KLEIN.539. A. R. is supported by the Dutch Research Council through grants OCENW.KLEIN.539 and VI.Vidi.203.045.}
\author{Benjamin Jany$^2$}
\address{$^2$University of Kentucky, USA.}

\author{Alberto Ravagnani$^1$}

\usepackage{setspace}
\begin{document}

\maketitle

\begin{abstract}
We initiate the study of the duality theory of locally recoverable codes, with a focus on the applications. We characterize the locality of a code in terms of the dual code, and introduce a class of invariants that refine the classical weight distribution. In this context, we establish a duality theorem analogous to (but very different from) a MacWilliams identity. As an application of our results, we obtain two new bounds for the parameters of a locally recoverable code, including an LP bound that improves on the best available bounds in several instances.
\end{abstract}


\medskip

\section{Introduction}

The ever increasing amount of data stored online
requires more and more efficient 
distributed storage and retrieval systems.
In this context, locally recoverable codes (LRC) have emerged as a solution that allows
to repair a single erasure by accessing a small set of servers,
but also to repair several erasures simultaneously, if needed.
LRC codes have been extensively studied in the last decade; see~\cite{tamo2014family,Gopalan,papailiopoulos2014locally,tamo2016optimal,Westerback,han20,guruswami19,hao20,Micheli,goparaju2014binary,silberstein2015optimal,silberstein2013optimal} among many others.

The parameters of an LRC code reflect those of the corresponding distributed storage system. If the code $\mC$ has alphabet $\F_q$ (the finite field with $q$ elements), length $n$, dimension $k$, minimum distance $d$, and locality $r$, then the corresponding system has $n$ severs, can store $k$ symbols, can recover each code symbol by contacting~$r$ servers, and can repair a total of $d-1$ erasures.

A natural and to date still wide open question asks to
describe how the parameters of an LRC code relate to each other. The most famous result in this context is Theorem~\ref{thm:single} below, which 
expresses a lower bound for the number of servers for fixed parameters $(k,d,r)$.
Several other bounds have been established; see~\cite{hu16,cadambe2015bounds, tamo2016bounds, hao2020bounds, agarwal2018combinatorial, chen2020improved, wang2019bounds, balaji2017bounds,guruswami2019long,tamo2016cyclic,fang2022bounds}.

In this paper, we initiate the study of how the parameters of an LRC code relate to the parameters of its dual code. The main result that we obtain in this context is a MacWilliams-type identity for a set of invariants that refine the weight distribution of a code, and that capture the concept of locality; see Theorem~\ref{thm:dualitywij}.
While we call our result a MacWilliams-type identity, it does not fit under the classical framework
of MacWilliams identities based on group partitions; see e.g.~\cite{gluesing2015fourier}.

The focus of this paper is on the concrete applications of the mentioned duality theory for LRC codes. More precisely, we apply our results to obtain two bounds for the parameters of an LRC code.
The first bound relates the locality of a code to its length, dimension, minimum distance, dual distance, and field size. To our best knowledge, this is the first bound that attempts to describe the connections among these parameters. The second bound is an LP-type result, which is based on the MacWilliams-type identities we mentioned earlier. The bound improves on the best available bounds for the dimension of an LRC code of given parameters $(q,n,d,r)$ in several instances.

\section{Preliminaries}
Throughout this paper, $q$ is a prime power, $\F_q$ is the finite field of $q$ elements, and $n \ge 2$ is an integer.  We start by recalling some notions from classical coding theory; see for instance~\cite{macwilliams1977theory}.

\begin{definition}
A (\emph{linear}) \emph{code} 
is a non-zero $\F_q$-linear subspace~$\mC \le \F_q^n$.  
The \emph{minimum} (\emph{Hamming}) \emph{distance} of $\mC$ is 
$$\dH(\mC) = \min\{\wH(x) \mid x \in \mC, \; x \ne 0\},$$
where $\wH$ denotes the (Hamming) weight on $\F_q^n$. 
\end{definition}

It is well-known that the $\F_q$-dimension $k$ of a code $\mC$ satisfies $k \le n- \dH(\mC)+1$. This inequality is the famous \emph{Singleton Bound} and codes meeting it with equality are called \emph{MDS} codes. 

An LRC code is a linear code for which an additional parameter called \emph{locality} is considered. The latter is defined as follows.

\begin{definition}\label{def:locality}
A code $\mC \le \F_q^n$
has \emph{locality} $r$ if for every $i \in \{1,...,n\}$ there exists a set $S_i$, called a \emph{recovery set} for (coordinate) $i$, with the following properties: 
\begin{enumerate}
 \item $i \notin S_i$,
 \item $|S_i| \le r$,
 \item if $x,y \in \mC$
 and $x_j=y_j$ for all $j \in S_i$, then $x_i=y_i$.
\end{enumerate}
\end{definition} 
Given a recovery set $S_i$ for $i$, it is possible to reconstruct the coordinate $x_i$ of any codeword $x \in \mC$ using a \emph{recovery function} $f_i : \pi_{S_i}(\mC) \rightarrow \F_q$, where $\smash{\pi_{S_i}: \F_q^n \to \F_q^{|S_i|}}$ is the projection onto the coordinates indexed by the elements of~$S_i$.
More formally, $f_i(y)$ is the only field element $\alpha \in \F_q$ for which there exists $x \in \mC$ with $x_i=\alpha$ and $\pi_{S_i}(x)=y$.

Recent work from \cite{tan2021minimum} introduced the notion of linear locality. The latter is defined similarly as in Definition~\ref{def:locality}, with the additional requirement that a recovery set $S_i$ for $i$ must repair coordinate $i$ via a linear recovery function. However, as the following proposition shows, if $\mC$ is a linear code, recovery functions must always be linear. Hence for linear codes, the notion of linear locality is equivalent to that of locality. 

\begin{proposition}\label{thm:linearfct}
Let $\mC \le \F_q^n$ be a linear code and $S_i$ be a recovery set for the coordinate $i$ with recovery function $f_i$. Then $f_i$ is an $\F_q$-linear map. 
\end{proposition}

\begin{proof}
Suppose $S_i := \{1, \ldots, r\}$ and $i > r$ without loss of generality. Let $x, y \in \pi_{S_i}(\mC)$ and $\alpha \in \F_q$. Using the properties of the recovery set $S_i$, it is easy to see that $x \in \pi_{S_i}(\mC)$ if and only if $(x, f_i(x)) \in \pi_{S_i \cup \{i\}}(\mC)$. Since $\pi_{S_i \cup \{i\}}(\mC)$ is a linear code, $(x + \alpha y, f_i(x) + \alpha f_i(y)) \in \pi_{S_i \cup \{i\}}(\mC)$. This implies that $f_i(x + \alpha y) = f_i(x) + \alpha f_i(y)$, which concludes the proof. 
\end{proof}

It is a classical problem in coding theory to understand the different trade-offs between code parameters.
The most famous of these problems asks to determine the largest dimension of a linear code defined over $\F_q$ with fixed length $n$ and minimum distance~$d$. An analogous question arises in the study of LRC codes where locality is considered as well. In~\cite{Gopalan}, the authors showed how locality impacts the maximum dimension of an LRC code by establishing a generalization of the Singleton Bound. The bound reads as follows.

\begin{theorem}[Generalized Singleton Bound] \label{thm:single}
Let $\mC \le \F_q^n$ be a code with locality $r$, dimension $k$ and minimum distance~$d$. Then 
\begin{equation}\label{eqt:single}
    k + \left \lceil \frac{k}{r} \right \rceil \leq n - d+2.
\end{equation}
\end{theorem}

Note that the bound of Theorem~\ref{thm:single} coincides with the classical Singleton Bound if~$k=r$. Codes whose parameters meet the bound (\ref{eqt:single}) with equality are called \emph{optimal LRC} codes. In~\cite{tamo2014family} it has been shown that the bound is tight when $q \geq n$, $r\mid k$ and $r+1 \mid n$.
However, for $q < n$, or if the divisibility constraints are not satisfied,
similarly to MDS codes, optimal LRC codes do not always exist. In order to determine the existence of optimal LRC codes for small~$q$, it has therefore become of interest to establish bounds dependent on both the locality and the underlying field size of a code. The following is a shortening bound established in~\cite{cadambe2015bounds}, that improves the Singleton-type bound in Theorem~\ref{thm:single}. In the minimum we include the (trivial) case $t=0$, even though the original statement doesn't. Note that for some parameters the minimum is indeed attained by $t=0$.

\begin{theorem}\label{thm:kopt}
Let $\mC \le \F_q^n$ be a code with locality $r$, dimension $k$ and minimum distance $d$. We have
\begin{equation}\label{eqt:kopt}
    k \le \min_{\substack{t \in \Z \\ t\ge 0}}\left\{rt+k_{\textnormal{opt}}^{(q)}(n-t(r+1),d)\right\},
\end{equation}
where $k_{\textnormal{opt}}^{(q)}(n,d)$ is the largest possible dimension of a code of length $n$ and minimum distance $d$ over $\F_q$.
\end{theorem}

Although the bound of Theorem~\ref{thm:kopt} is a refinement of Theorem~\ref{thm:single} for all parameter sets there is a ``computational'' drawback to this refinement.
Indeed, determining the value of~$k_{\textnormal{opt}}^{(q)}(n,d)$ for given $d, n, q$ is still a wide open problem in classical coding theory, which makes the RHS of \eqref{eqt:kopt} difficult to evaluate.

In this paper, we derive two bounds for the parameters of LRC codes that
are based on their duality theory
and improve on 
the state of the art.
Note that the connection
between the locality of a code and that of its dual code is in general unclear and quite unexplored: a secondary goal of this paper is also to start filling this important gap.

\section{A Duality Theory for Codes with Locality} \label{sec:dual}

In this section we start developing 
a duality theory for codes with locality. A more extensive treatment of this topic is anticipated in an extended version of this work.

The results established in this section will allow us to derive two bounds for the parameters of an LRC code based on duality arguments. We will present the bounds later in Section~\ref{sec:appl}, while this section  focuses on the theoretical aspects.

We start by explaining how the locality of a code $\mC \leq \F_q^n$ can be characterized in terms of the support of codewords in the dual code $\mC^{\perp}$. Recall that for a vector $x \in \F_q^n$, its \emph{support} is defined as $\sigma(x) := \{ i \mid x_i \neq 0\}$. A code $\mC$ is said to be \emph{non-degenerate} if for all $i \in [n]$ there exist $x \in \mC$ such that $i \in \sigma(x)$. Using the notion of support, we obtain the following equivalent characterization of the locality parameter, which we use extensively in our approach.

\begin{lemma}\label{local/sup}
Let $r \ge 1$ be an integer. A linear code $\mC \le \F_q^n$ has locality $r$ if and only if for any $i \in \{1, \dots, n\}$ there exists $x \in \mC^{\perp}$ with $i \in \sigma(x)$ and $\wH(x) \le r+1$.
\end{lemma}

\begin{proof}
Let $S := \{1, \ldots , r\}$ be a recovery set for the coordinate $r+1$ with recovery function $f$, i.e., $f(v_1, \ldots, v_r) = v_{r+1}$ for all $v \in \C$. Since $f$ is linear by Proposition~\ref{thm:linearfct}, for all~$v \in \mC$ we have~$\smash{f(v_1, \ldots, v_r) = \sum_{i=1}^r \lambda_i v_i= v_{r+1}}$, where $\lambda_i \in \F_q$ for all $1 \leq i \leq r$. This equation shows that~$w:=\smash{(\lambda_1, \ldots, \lambda_r, -1, 0 , \ldots , 0)} \in \C^{\perp}$, as desired.
Conversely, it is easy to show that if $w \in \mC^{\perp}$ has the above form then it induces a recovery function for the coordinate~$r+1$.
\end{proof}

Lemma~\ref{local/sup} was first established in  \cite[Lemma 5]{guruswami19}. However, the proof of the result strongly relies on the fact that the recovery functions are always linear, which, to the best of our knowledge, is first explicitly stated in this paper; see 
Proposition~\ref{thm:linearfct}. 

We continue 
by introducing a new set of invariants for an LRC code.
Recall that the
\emph{weight distribution}
of a code
$\mC \le \F_q^n$ is the tuple $(W_0(\mC),\dots,W_n(\mC))$, where $W_i(\mC)$ denotes the number of codewords of weight $i$ in $\mC$.
In this paper, we introduce and investigate the following related
notion of weight distribution, that will enable us to capture the locality of codes through their dual.

\begin{notation}
Let $\mC \le \F_q^n$ be a  non-degenerate code. For $1 \le j \le n$ and $2 \le i \le n$, define $W_i^j(\mC)$ as follows:
\begin{align*}
    W_i^j(\mC) = |\{x \in \mC \mid \wH(x)=i, \, j \in \sigma(x)\}|.
\end{align*}
\end{notation}

Lemma~\ref{local/sup} 
implies the following characterization of the locality of a code.

\begin{proposition}
A code $\mC \le \F_q^n$ has locality $r$ if and only if $ \smash{\sum_{i=2}^{r+1}W_{i}^j(\mC^\perp) > 0}$ for all $1 \le j \le n$.
\end{proposition}

A preliminary, small result that will be needed later is the following.

\begin{lemma} \label{lem:helpdual}
Let $\mC \le \F_q^n$ be non-degenerate. For all $1 \le i \le n$ and all $1\le j \le n$ we have
$$W_i(\mC) = \frac{\textstyle\sum_{j=1}^n W_i^j (\mC)}{i}.$$
\end{lemma}
\begin{proof}
We prove the statement by counting the elements in the set $$S=\{(x,j) \mid x \in \mC, \, \wH(x) = i, \, j \in \sigma(x)\}$$ in two ways. Firstly, we have
\begin{align*}
    |S| = \displaystyle\sum_{j=1}^n|\{x \in \mC \mid \wH(x) = i, j \in \sigma(x)\}| = \displaystyle\sum_{j=1}^n W_i^j(\mC).
\end{align*}
On the other hand, 
\begin{align*}
    |S| = \sum_{\substack{x \in \mC \\ \wH(x)=i}}|\{j \mid  j \in \sigma(x)\}| = W_i(\mC) \cdot i,
\end{align*}
concluding the proof.
\end{proof}

For the main result of this section we will need the celebrated MacWilliams identities, which give a closed formula for the weight distribution of a code in terms of the weight distribution of its dual; for a reference see \cite[Theorems~7.1.3 and~7.2.3]{huffman_pless_2003}. 

\begin{theorem}[MacWilliams identities] \label{thm:mw}
Let $\mC \le \F_q^n$ be a code. For all $0 \le i \le n$ we have  
$$W_i(\mC^\perp) = \frac{1}{|\mC|} \sum_{j=0}^n  K_{i,j,q}W_j(\mC),$$
where $$K_{i,j,q} = \sum_{t=0}^i\binom{n-t}{i-t}\binom{n-j}{t}(-1)^{i-t}q^t$$
denotes the \emph{Krawtchouk coefficients}.
\end{theorem}

With the aid of Theorem~\ref{thm:mw} we are now able to show that
for a code $\mC \le \F_q^n$ and for all $1\le i,j \le n$, the value of~$W_i^j(\mC)$ is fully determined by the set $\{W_a^b(\mC^\perp) \mid 1 \le a,b \le n\}$.

\begin{theorem} \label{thm:dualitywij}
Let $\mC \le \F_q^n$ be a  non-degenerate code. For all $1 \le i \le n$ and all $1\le j \le n$ we have 
\begin{multline*}
    W_i^j(\mC) =
    \frac{1}{|\mC^\perp|}\Bigg(K_{i,0,q}\left(1-\frac{1}{q}\right)- \frac{K_{i,1,q}}{q}\left(q-1+W_{2}^j(\mC^\perp)\right) + \\
    \sum_{s=2}^n K_{i,s,q} \Bigg[ \displaystyle\sum_{t=1}^n \frac{W_{s}^t (\mC^\perp)}{s}-  \frac{1}{q} \Bigg( (q-1)\left({\displaystyle\sum_{t=1}^n \frac{W_{s-1}^t (\mC^\perp)}{s-1}} -W_{s-1}^j(\mC^\perp) \right) + \\
 \displaystyle\sum_{t=1}^n \frac{W_{s}^t (\mC^\perp)}{s}+   W_{s+1}^j(\mC^\perp) + (q-2) W_s^j(\mC^\perp) \Bigg) \Bigg] \Bigg),
\end{multline*}
where we set $W_{n+1}^j(\mC)=0$ for all $1 \le j \le n$ by convention.
\end{theorem}

\begin{proof}
Let $S_j=[n] \backslash \{j\}$ for all $j$. First note that
\begin{align*}
    W_i ^j(\mC) &= W_i(\mC)-|\{x \in \mC \mid \wH(x)=i, \, \sigma(x) \subseteq S_j\}| \\
    &=W_i(\mC)-W_i(\mC(S_j)) \\
    &=\frac{1}{|\mC^\perp|} \sum_{s=0}^n K_{i,s,q} W_s(\mC^\perp) -  \frac{1}{|\mC(S_j)^\perp|} \sum_{s=0}^n K_{i,s,q} W_s(\mC(S_j)^\perp),
\end{align*}
where the latter equality comes from the MacWilliams identities (Theorem~\ref{thm:mw}).
We will look at $W_s(\mC(S_j)^\perp)$ more in detail, starting with the observation that, since $\mC^\perp$ is non-degenerate, we have $\mC(S_j)^\perp = \langle e_j \rangle \oplus \mC^\perp$, where $e_j$ denotes the $j$th unit vector in $\F_q^n$.
Therefore, $W_s(\mC(S_j)^\perp) = W_s(\langle e_j \rangle \oplus \mC^\perp)$. Every element in $\langle e_j \rangle \oplus \mC^\perp$ can be written uniquely as $\alpha e_j + x$ where $\alpha \in \F_q$ and $x \in \mC^\perp$. Note that for $\alpha \in \F_q$ and $x \in \mC^\perp$ we have
\begin{align*}
\wH(\alpha e_j + x)  =    
\begin{cases}
\wH(x)+1 \quad &\textnormal{ if $\alpha \ne 0, \, x \in \mC^\perp(S_j)$,} \\
\wH(x) \quad &\textnormal{ if $\alpha = 0$,} \\
\wH(x)-1 \quad &\textnormal{ if $\alpha \ne 0, \, x_j=-\alpha$,} \\
\wH(x) \quad &\textnormal{ if $\alpha \ne 0, \, x_j \notin \{-\alpha,0\}$.} \\
\end{cases}
\end{align*}
By convention, we set $W_{n+1}^t=0$ for all $1 \le t \le n$. Then, for $1 \le s \le n-1$, straightforward computations (together with Lemma~\ref{lem:helpdual})
 give that
\begin{align*}
    W_s(\mC(S_j)^\perp) = W_s(\langle e_j \rangle \oplus \mC^\perp) = 
 A_s^1 + A_s^2 + A_s^3 +A_s^4,
\end{align*}
where
\allowdisplaybreaks
\begin{align*}
    A_s^1 &= |\{(\alpha,x) \mid \alpha \in \F_q^*, \, x \in \mC^\perp(S_j), \, \wH(x)=s-1\}| \\
    &= |\F_q^*|\cdot|W_{s-1}(\mC^\perp(S_j))| = (q-1) \left(W_{s-1}(\mC^\perp)-W_{s-1}^j(\mC^\perp)\right) \\
    &= \begin{cases}
    (q-1) \quad \textnormal{if $s=1$,} \\
    (q-1)\left({\displaystyle\sum_{t=1}^n W_{s-1}^t (\mC^\perp)}/{(s-1)} -W_{s-1}^j(\mC^\perp)\right) \textnormal{else,}
    \end{cases} \\
    A_s^2 &= |\{(\alpha,x) \mid \alpha=0, x \in \mC^\perp, \wH(x)=s\}| = W_s(\mC^\perp) = {\sum_{t=1}^n W_{s}^t (\mC^\perp)}/{s},\\
    A_s^3 &= |\{(\alpha,x) \mid \alpha \in \F_q^*, \, x \in \mC^\perp,  x_j = -\alpha,   \wH(x)=s+1\}| = W_{s+1}^j(\mC^\perp), \\
    A_s^4 &= |\{(\alpha,x) \mid \alpha \in \F_q^*, \, x \in \mC^\perp, x_j \ne -\alpha,  \wH(x)=s\}| = (q-2) W_s^j(\mC^\perp).
 \end{align*}
It is easy to see that  $W_n(\langle e_j \rangle \oplus \mC^\perp) = A_n^1+A_n^2+A_n^4$. Therefore we have derived a closed formula for $W_s(\mC(S_j)^\perp)$ for all $1 \le s \le n$. We finally obtain
\begin{multline*}
    W_i ^j(\mC) = \frac{1}{|\mC^\perp|}\left(K_{i,0,q} + \sum_{s=1}^n  K_{i,s,q} A_s^2\right)- \\ \frac{1}{|\mC(S_j)^\perp|} \left(K_{i,0,q}+ \sum_{s=1}^n  K_{i,s,q} \Bigg( A_s^1 + A_s^2 + A_s^3 +A_s^4\Bigg)\right),
\end{multline*}
from which the formula of the theorem follows (using also that $|\mC(S_j)^\perp| = q^{n-k+1}$, which holds because $\mC$ is non-degenerate by assumption).
\end{proof}

Theorem~\ref{thm:dualitywij} may be 
regarded as a MacWilliams-type identity
for the $W_i^j(\mC)$'s, even though these numbers do not naturally fit into the framework of the MacWilliams identities (they do not represent the cardinalities of the blocks of a partition of the underlying code).
One may say that, as invariants, the classical weight distribution is to the minimum distance what the $W_i^j(\mC)$'s are to the minimum distance \emph{and} locality.

\begin{remark} \label{rem:weightdistrans}
For a non-degenerate code $\mC \le \F_q^n$ with $d(\mC) \ge 2$ and an integer $1 \le j \le n$, consider a codeword $x \in \mC^\perp$ with $x_j \ne 0$. Then the tuple $$\left(\frac{W_2^j(\mC^{\perp})}{q-1},\frac{W_3^j(\mC^{\perp})}{q-1},\dots,\frac{W_n^j(\mC^{\perp})}{q-1}\right)$$ is the weight distribution of the (non-linear) code $\mC^\perp(S_j)+x$, where $S_j=[n] \setminus \{j\}$. Therefore Theorem~\ref{thm:dualitywij} shows that the weight distribution of $\mC^\perp(S_j)+x$ is fully determined by the weight distribution of the translates $\{\mC(S_i)+x^i \mid i \in [n]\}$, where for all $i \in [n]$, $x^i$ is any vector in $\mC$ with $x_i^i \ne 0$.
\end{remark}

\section{Applications} \label{sec:appl}

In this section we investigate two concrete applications of the duality theory of LRC codes. More precisely, we 
use duality arguments to
establish two bounds for the parameters of LRC codes, which often improve on the previously known best bounds.

\subsection{A Dual Distance Bound}

We start with a result that relates the parameters of an LRC code to its dual distance. 
The following result holds true for any linear code.

\begin{proposition}\label{prop:weightbound}
Let $\mC \le \F_q^n$ be a code of dimension $k \ge 2$ and minimum distance $d$. Let $i \in \{1, \ldots, n\}$. If there exists $x \in \mC$ such that $i \in \sigma(x)$, then there exists $y \in \mC$ such that $i \in \sigma(y)$ and 
$$\omega^\HH(y) \leq n- k +1 - \frac{(d-q)}{q}.$$ 
\end{proposition}

\begin{proof}
Assume $i=1$ without loss of generality. Fix $x \in \mC$ of minimum weight, say $m$, such that $1 \in \sigma(x)$. Without further loss of generality, we can assume $x=(1, \ldots,1,0, \ldots, 0)$. 
We have $m \le n-k+1$.  Since~$\mC$ has dimension $k$, there exists a non-zero $y \in \mC$ with $y_j=0$ for $j \in \{1\} \cup \{n, n-1, \ldots, n-k+3\}$. Let $S=\{1, \ldots, m\}$. Since~$\mC$ has minimum distance $d$,
$\pi_S(y)$ has weight at least $d-(n-m)+k-2=d-n+m+k-2$. Fix $\alpha \in \F_q^*$ with
\begin{align*}
    \{j \mid \pi_S(y)_j = \alpha\} \ge \frac{\wH(\pi_S(y))}{(q-1)} \ge \frac{(d-n+m+k-2)}{(q-1)}.
\end{align*}
Then $x':=\alpha x - y$ has $x'_1 \neq 0$ and weight at most
\begin{multline*}
    m-\frac{d-n+m+k-2}{q-1} + (n-m)-(k-2) =
    n-k+2 - \frac{d-n+m+k-2}{q-1}.
\end{multline*}
Therefore, by the choice of $m$, it must be that
$$m \le n-k+2 - \frac{d-n+m+k-2}{q-1},$$
which is the same as
$m \le n-k+1 - (d-q)/q$.
This shows the desired result.
\end{proof}

Note for linear codes with $d>q$,
Proposition~\ref{prop:weightbound} improves on the 
Singleton Bound.
In particular, it shows that an MDS code of dimension $k \ge 2$ always has $d \ge q$.
This fact is often shown by imposing that the number of codewords of weight $d+1$ in an MDS code is non-negative; see e.g.~\cite[Corollary~7.4.3]{huffman_pless_2003}.

By applying Proposition~\ref{prop:weightbound} to the dual $\mC^{\perp}$ of an LRC code with locality $r$, and by using Theorem~\ref{thm:single}, we get the following bound. To our best knowledge, this is the first bound for LRC codes that takes into account the six parameters $q$, $n$, $k$, $d$, $d^\perp$, $r$.

\begin{corollary}\label{cor:locdualdist}
Let $\mC \le \F_q^n$ an LRC code of dimension $2 \leq k \leq n-2$, minimum distance $d$, and locality $r$. Moreover, \\let~$d^\perp :=\dH(\mC^\perp)$. We have
\begin{align*}
    d^\perp +\frac{d^\perp-q}{q} \le n-d+3-\left\lceil \frac{k}{r}\right\rceil.
\end{align*}
\end{corollary}

\subsection{An LP Bound for LRC Codes}

Delsarte’s linear programming bound is 
a 
renownedly powerful
tool to estimate the size of codes of a certain length and minimum distance.
It mainly relies on the Macwilliams identities (Theorem~\ref{thm:mw}). In this section we study a new LP bound, with focus on LRC codes.
We use our main duality result,
Theorem~\ref{thm:dualitywij},
to construct a linear program.
For ease of exposition, we introduce the following notation.

\begin{notation} \label{not:perp}
Let $A=\{a_{ij} \mid 1\le i \le n+1, \, 1\le j \le n\}$ be non-zero integers. For $1 \le i,j \le n$ we denote by $a_{ij}^\perp$ the following linear combination of elements in $A$:
\begin{multline*}
    a_{ij}^\perp =
    K_{i,0,q}\left(1-\frac{1}{q}\right)- \frac{K_{i,1,q}}{q}\left(q-1+a_{2j}\right)+ \\ \sum_{s=2}^n K_{i,s,q} \Bigg[ \displaystyle\sum_{t=1}^n \frac{a_{st}}{s}- \frac{1}{q} \Bigg( (q-1)\left({\displaystyle\sum_{t=1}^n \frac{a_{s-1,t}}{s-1}} -a_{s-1,j} \right) +\\
 \displaystyle\sum_{t=1}^n \frac{a_{st}}{s}+   a_{s+1,j} + (q-2) a_{sj} \Bigg) \Bigg].
\end{multline*}
\end{notation}

With Notation~\ref{not:perp} in mind, we can set up the following linear program that will give bounds for codes with locality. 

\begin{theorem}[LP bound for LRC codes] \label{thm:lpbound}
Let $\mC \le \F_q^n$ be a non-degenerate LRC code of minimum distance at least $d$, dimension $k$, and locality $r$. 
Let $\mu^*$ denote the minimum
value of 
$$\sum_{i=1}^{n+1}\left({\displaystyle\sum_{j=1}^{n} a_{ij}}/{i}\right),$$
where $a_{ij} \in \R$, for $1 \le i \le n+1$ and $1 \le j \le n$,
satisfy the following constraints:
\begin{itemize}
    \vspace{0.15cm}\item[(i)]  $a_{ij}\ge 0$ for $1 \le i,j \le n$, 
    \vspace{0.15cm}\item[(ii)]  $a_{ij}^\perp\ge 0$ for $1 \le i,j \le n$, 
    \vspace{0.15cm}\item[(iii)] $a_{ij}^\perp =0$ for $1 \le i \le d-1$ and $1 \le j \le n$,   
    \vspace{0.15cm}\item[(iv)]  $\smash{\textstyle\sum_{i=1}^{r+1}a_{ij} \ge q-1}$ for $1 \le j \le n$,
    \vspace{0.15cm}\item[(v)] $a_{1j} =0$ for $1 \le j \le n$, 
    \vspace{0.15cm}\item[(vi)]  $a_{n+1,j}=0$ for $1 \le j \le n$,
\end{itemize}
Then
$$k \le n- \lceil \log_q(1+\mu^*) \rceil.$$
\end{theorem}
\begin{proof}
We claim that for any non-degenerate linear code $\mC \le \F_q^n$ of minimum distance $d$ and locality $r$, the assignment $\smash{a_{ij} = W_i^j(\mC^\perp)}$ is a feasible solution of the linear program. Indeed, (i) and~(ii) are satisfied trivially. Constraint~(iii) guarantees that the minimum distance of the code is at least~$d$ and constraint~(iv) guarantees that the code has locality~$r$. Finally constraint~(v) makes sure that $\mC$ is non-degenerate and constraint~(vi) is needed by convention (in order to be able to apply Theorem~\ref{thm:dualitywij}). Therefore the optimum of the linear program gives a lower bound on the size of $\mC^\perp$ by Lemma~\ref{lem:helpdual}.
\end{proof}

\begin{remark}
    A different LP-type bound was established in~\cite{hu16}. However, the approach
    of~\cite{hu16} only applies to codes whose recovery sets are disjoint and of equal size, requiring $r+1$ to divide $n$ (see~\cite[Section IV]{hu16}).
    Those assumptions drastically reduce the parameters for which the bound can be applied and the bound of~\cite{hu16} is therefore not comparable with our LP bound.
\end{remark}

\begin{table}[h!]
\caption{$q=2$}
    \centering
    \renewcommand\arraystretch{1.2}
\begin{tabular}{|c|c|c|c|c|c|c|} 
 \hline
$n$ & $d$ & $r$ & LP & SH with LP  & SH exact & gen. Singl.
  \\\noalign{\global\arrayrulewidth 1.8pt}
    \hline
    \noalign{\global\arrayrulewidth0.4pt}
10 & 4 & 2 & $k \leq \red{4}$ & $ k \le 5$ & $k \le 5$ & $k \le 5$\\
\hline
11 & 3 & 3 & $k \le \red{6}$ & $k \le 7$ & $k \le 7$ & $k \le 7$\\
\hline
12 & 3 & 3 & $k \le \red{6}$ & $k \le 7$ & $ k \le 7$ & $k \le 7$\\
\hline
14 & 4 & 6 & $k \le \red{8} $ & $k \le 9$ & $k \le 9$ & $k \le 10$\\
\hline 
17 & 7 & 2 & $k \le \red{5}$ & $k \le 6$ & $ k \le 6$ & $k \le 8$ \\
\hline
18 & 8 & 2& $k\le \red{5}$ & $k \le 6$ & $k \le 6$ & $k \le 8$\\
\hline
20 & 9 & 5 & $k \le 5$ & $k \le 6$ & $k \le 5$ & $k \le 11$\\
\hline 
20 & 11 & 2 & $k \le \red{2}$ & $k \le 3$ & $k \le 3$ & $k \le 7$\\
\hline 
\end{tabular}
\label{table_LPbound1}
\end{table}

\begin{table}[h!]
\caption{$q=3$}
    \centering
    \renewcommand\arraystretch{1.2}
\begin{tabular}{|c|c|c|c|c|c|c|} 
 \hline
$n$ & $d$ & $r$ & LP & SH with LP  & SH exact & gen. Singl.
  \\\noalign{\global\arrayrulewidth 1.8pt}
    \hline
    \noalign{\global\arrayrulewidth0.4pt}
10 & 2 & 4 & $k \le \red{7}$ & $k \le 8$ & $k \le 8$ & $k \le 8$\\
\hline
11 & 6 & 4 & $k \le \red{4}$ & $k \le 5$ & $k \le 5$ & $k \le 5$\\
\hline
11 & 5 & 5 & $k \le \red{5}$ & $k \le 6 $ & $k \le 6$ & $k \le 6$\\
\hline 
13 & 9 & 2 & $k \le \red{2}$ & $k \le 3$ & $ k \le 3$ & $k \le 4$\\
\hline
14 & 2 & 6 & $k \le \red{11}$ & $k \le 12$ & $k \le 12$ & $k \le 12$\\
\hline
14 & 9 & 3 & $k \le 3$ & $k \le 4$ & $k \le 3$ & $k \le 5$\\
\hline
18 & 6 & 5 & $k \le 10$ & $k \leq 11$ & $k \le 10 $ & $k \le 11$\\
\hline
25 & 5 & 5 & $k \le \red{15}$ & $k \le 17$ & $k \le 17$ & $k \le 18$\\
\hline
\end{tabular}
\label{table_LPbound2}
\end{table}

\begin{table}[h!]
\caption{$q=4$}
    \centering
    \renewcommand\arraystretch{1.2}
\begin{tabular}{|c|c|c|c|c|c|c|} 
 \hline
$n$ & $d$ & $r$ & LP & SH with LP  & SH exact & gen. Singl.
  \\\noalign{\global\arrayrulewidth 1.8pt}
    \hline
    \noalign{\global\arrayrulewidth0.4pt}
    9 & 3 & 3 & $k \le \red{5}$  &$k \le 6$ & $k \le 6$ & $k \le 6$\\
    \hline
    10 & 2 & 4 & $k \le \red{7}$ & $k \le 8$ & $k \le 8$ & $k \le 8$\\
    \hline
    11 & 8 & 2 & $k \le 2$ & $k \le 3$ & $k \le 2$ & $k \le 3$\\
    \hline
    12 & 2 & 5 & $k \le \red{9}$ & $k \le 10$ & $ k \le 10$ & $ k\le 10$\\
    \hline
    14 & 2 & 6 & $k \le \red{11}$ & $k \le 12$ & $k \le 12$& $k \le 12$\\
    \hline 
    15 & 10 & 4 & $k \le 4$ & $k \le 5$ & $k \le 4$ & $k \le 5$\\
    \hline
   15 & 8 & 6& $k \le 6$ & $k \le 7$ & $k \le 6$ & $ k \le 7$\\
   \hline
    20 & 2 & 9 & $k \le \red{17}$ & $k \le 18$ & $k \le 18$ & $k \le 18$\\
    \hline
\end{tabular}
\label{table_LPbound3}
\end{table}

\begin{table}[h!]
\caption{$q=5$}
    \centering
    \renewcommand\arraystretch{1.2}
 \begin{tabular}{|c|c|c|c|c|c|c|} 
 \hline
$n$ & $d$ & $r$ & LP & SH with LP & SH exact & gen. Singl.
  \\\noalign{\global\arrayrulewidth 1.8pt}
    \hline
    \noalign{\global\arrayrulewidth0.4pt}
    9 & 3 & 3 & $k \le \red{5}$ & $k \le 6$ & $k \le 6$  & $k \leq 6$\\
    \hline
    11 & 3 & 4 & $k \le \red{7}$ & $k \le 8$ & $k \le 8$  & $k \leq 8$\\
    \hline
    14 & 2 & 6 & $k \le \red{11}$ & $k \le 12$ & $ k\le 12 $ & $k \le 12$\\
    \hline
    16 & 2 & 7 & $k \le \red{13}$ & $k \le 14$ & $k \le 14 $ & $k \le 14$\\
    \hline 
    16 & 2 & 7 & $k \le \red{13} $& $k \le 14$ & $k \le 14$ & $k \le 14$\\
    \hline
    18 & 2& 8 & $k \le \red{15}$ & $k \le 16$ & $k \le 16 $ & $k \le 16$\\
    \hline
    22 & 2 & 10 & $k \le \red{19}$ & $k \le 20$ & $k \le 20 $ & $k \le 20$\\
    \hline
    24 & 2 & 11 & $k \le \red{21}$ & $k \le 22$ & $k \le 22$ & $k \le 22$\\
    \hline
\end{tabular}
\label{table_LPbound4}
\end{table}

In Tables~\ref{table_LPbound1}-\ref{table_LPbound4} we present some results obtained 
with the LP bound of Theorem~\ref{thm:lpbound}. 
For fixed parameters $(q,n,d,r)$, we give
an upper bound on the dimension $k$ of a code with those parameters.  
We compare our linear programming bound (denoted by ``LP'') with the bound of Theorem~\ref{thm:kopt} (denoted by ``SH'', for ``shortening'') and with the generalized Singleton Bound of Theorem~\ref{thm:single}.

Note that the bound of Theorem~\ref{thm:kopt} relies on the value of the maximum dimension of a code with prescribed length and minimum distance, $k_{\textnormal{opt}}^{q}(n,d)$. In the tables, we consider both 
estimates for these quantities given by  Delsarte's LP bound as well as their exact values. Since the latter can be computed only for very small parameters, the most fair comparison is the one where these quantities are estimated, rather than computed. Nonetheless, the tables below show that our LP bound often improves on 
the bound of Theorem~\ref{thm:kopt} even when we use the \textit{exact} value of 
$k_{\textnormal{opt}}^{q}(n,d)$
for computing it. We highlight
in {\color{red}red} when this happens. Although we present results for only four different field sizes due to spacing constraints, we additionally found  that the LP bound of Theorem~\ref{thm:lpbound} outperforms other bounds for larger field sizes $q$ as well.
The computations were performed using \texttt{SageMath}.

\section{Discussion and Future Work}
The results presented in this paper are first, important steps towards developing a complete duality theory of locally recoverable codes.
We have established a MacWilliams-type identity for a set of invariants that refine the weight distribution,
and that capture the locality of a code.

In this paper, we have focused on the applications of said duality theory. We have obtained two bounds for the parameters of a locally recoverable code. The first bound relates the locality of a code to the field size and the dual distance. The second bound uses an LP program and improves on the
best available bounds in several instances. 

An extended version of this paper is anticipated, where a more complete duality theory will be developed and yet other applications will be presented.

\bibliographystyle{ieeetr}
\bibliography{ourbib}
\end{document}